\documentclass[10pt,showpacs,letterpaper,aps,twocolumn,pra,superscriptaddress,nofootinbib]{revtex4-2}%,nofootinbib
\usepackage{relsize,natbib}

\usepackage{amsthm,amsmath,amssymb,graphicx,comment,braket}
\usepackage{bbold}
\usepackage[none]{hyphenat}

\usepackage{abraces}%For under and overbraces that aren't just curly brackets

\usepackage{microtype}

\usepackage{units}
\usepackage{multirow}
\usepackage[normalem]{ulem}
\usepackage{mathtools}
\DeclareMathSymbol{\shortminus}{\mathbin}{AMSa}{"39}

\usepackage[usenames,dvipsnames,table]{xcolor}
\definecolor{purple}{RGB}{128,0,128}
\definecolor{ultramarine}{RGB}{63, 0, 255}
\definecolor{medblue}{RGB}{0, 0, 100}
\definecolor{panblue}{RGB}{0,24,150}
\definecolor{carmine}{RGB}{150, 0, 24}
\definecolor{gray}{RGB}{150, 150, 150}
\definecolor{darkred}{RGB}{200, 0, 0}
\definecolor{darkgreen}{RGB}{0, 80, 0}
\definecolor{darkblue}{RGB}{0, 0, 200}
\definecolor{darkorange}{rgb}{1.0, 0.55, 0.0}
\definecolor{nred}{rgb}{0.9,0.1,0.1}
\definecolor{nblack}{rgb}{0,0,0}
\definecolor{nblue}{rgb}{0.2,0.2,0.8}
\definecolor{ngreen}{rgb}{0.2,0.6,0.2}
\definecolor{darkestblue}{rgb}{0, 0, 0.3}

\newcommand{\tr}[1]{\mathsf{tr}[#1]}
\newcommand{\herm}{\mathsf{Herm}(\mathcal{H})}
\newcommand{\Z}{\mathbb{Z}}

\usepackage{hyperref}
\hypersetup{colorlinks, breaklinks, linkcolor=darkred, urlcolor=darkblue, citecolor=darkgreen}

\usepackage{bm}

\makeatletter
\newcommand{\neutralize}[1]{\expandafter\let\csname c@#1\endcsname\count@}
\makeatother

\usepackage{dsfont}
\usepackage{yfonts}
\usepackage{rotating}
\usepackage{floatpag}
\rotfloatpagestyle{empty}

\usepackage{amsmath,amsthm,amssymb}
\usepackage{xspace,enumerate,color,epsfig}
\usepackage{graphicx}
\graphicspath{{.}{./figures/}}
\usepackage{marvosym}

\usepackage{stmaryrd}
\usepackage{docmute}
\usepackage{keycommand}

\usepackage{pifont}% http://ctan.org/pkg/pifont

\usepackage{enumitem}

\let\olddagger\dagger
\renewcommand{\dagger}{\ensuremath{\olddagger}\xspace}

\usepackage{makeidx}
\makeindex

\theoremstyle{plain}
\newtheorem*{main theorem}{Main Theorem}
\newtheorem{theorem}{Theorem}[]

\newtheorem{lemma}[theorem]{Lemma}

\newtheorem{example*}[theorem]{Example*}
\newtheorem{examples*}[theorem]{Examples*}

\newtheorem{remark*}[theorem]{Remark*}

\newtheorem*{search problem}{Search Problem}

\usepackage{color}
\def\bR{\begin{color}{red}}
\def\bB{\begin{color}{blue}}
\def\bM{\begin{color}{magenta}}
\def\bC{\begin{color}{cyan}}
\def\bW{\begin{color}{white}}
\def\bBl{\begin{color}{black}}
\def\bG{\begin{color}{green}}
\def\bY{\begin{color}{yellow}}
\def\e{\end{color}\xspace}
\newcommand{\bit}{\begin{itemize}}
\newcommand{\eit}{\end{itemize}\par\noindent}
\newcommand{\ben}{\begin{enumerate}}
\newcommand{\een}{\end{enumerate}\par\noindent}
\newcommand{\beq}{\begin{equation}}
\newcommand{\eeq}{\end{equation}\par\noindent}
\newcommand{\beqa}{\begin{align*}}
\newcommand{\eeqa}{\end{align*}}
\newcommand{\beqn}{\begin{align}}
\newcommand{\eeqn}{\end{align}\par\noindent}
\def\jR{\begin{color}{black}}
\def\jB{\begin{color}{black}}
\def\jM{\begin{color}{magenta}}
\def\jC{\begin{color}{cyan}}
\def\jW{\begin{color}{white}}
\def\jBl{\begin{color}{black}}
\def\jG{\begin{color}{green}}
\def\jY{\begin{color}{yellow}}

\usepackage{stmaryrd}

\allowdisplaybreaks

\begin{document}

\title{Uniqueness of Noncontextual Models for Stabilizer Subtheories}

\author{David Schmid}
%\email{dschmid@perimeterinstitute.ca}
\affiliation{Perimeter Institute for Theoretical Physics, 31 Caroline Street North, Waterloo, Ontario Canada N2L 2Y5}
\affiliation{Institute for Quantum Computing, University of Waterloo, Waterloo, Ontario N2L 3G1, Canada}
\affiliation{International Centre for Theory of Quantum Technologies, University of Gda\'nsk, 80-308 Gda\'nsk, Poland}
\author{Haoxing Du}
\affiliation{Department of Physics, University of California, Berkeley, Berkeley, California 94720, USA}
\author{John H. Selby}
\affiliation{International Centre for Theory of Quantum Technologies, University of Gda\'nsk, 80-308 Gda\'nsk, Poland}
\author{Matthew F. Pusey}
\affiliation{Department of Mathematics, University of York, Heslington, York YO10 5DD, United Kingdom}

\begin{abstract}

We give a complete characterization of the (non)classicality of all stabilizer subtheories. 
First, we prove that there is a {\em unique} nonnegative and diagram-preserving quasiprobability representation of the stabilizer subtheory in all odd dimensions, namely Gross's discrete Wigner function. This representation is equivalent to Spekkens' epistemically restricted toy theory, which is consequently singled out as the unique noncontextual ontological model for the stabilizer subtheory.
Strikingly, the principle of noncontextuality is powerful enough
(at least in this setting) to single out {\em one particular} classical realist interpretation. 
Our result explains the practical utility of Gross's representation by showing that (in the setting of the stabilizer subtheory) negativity in this particular representation implies generalized contextuality. Since negativity of this particular representation is a necessary resource for universal quantum computation in the state injection model, it follows that generalized contextuality is also a necessary resource for universal quantum computation in this model.
In all even dimensions, we prove that there does not exist any nonnegative and diagram-preserving quasiprobability representation of the stabilizer subtheory, and, hence, that the stabilizer subtheory is contextual in all even dimensions.  
\end{abstract}
\maketitle
%\tableofcontents
%\section{Introduction}

Quantum computers have the potential to outperform classical computers at many tasks. One of the major outstanding problems in quantum computing is to understand what physical resources are necessary and sufficient for universal quantum computation. These resources may depend on one's model of computation~\cite{PhysRevLett.91.147902,PhysRevA.72.042316,PhysRevLett.81.5672}, and in some cases it seems that neither entanglement nor even coherence is required in significant quantities~\cite{PhysRevA.72.042316}. 

The primary obstacle to building a quantum computer is that implementing low-noise gates is difficult in practice. 
% The standard way to combat the cascade of logical errors in quantum computer is fault-tolerant quantum computation, wherein one 
While there are no gate sets which are easy to implement and also universal~\cite{eastin2009restrictions}, the entire stabilizer subtheory~\cite{gottesman1997stabilizer,gottesman1998heisenberg} can in fact be implemented in a fault-tolerant manner  relatively easily. To promote the stabilizer subtheory to universal quantum computation, one must supplement it with additional nonstabilizer (or `magic') processes. Because these nonstabilizer resources do not have a straightforward fault-tolerant implementation, they are typically noisy. To get around this problem, Bravyi and Kitaev~\cite{bravyi2005universal} introduced the magic state distillation scheme, whereby fault-tolerant stabilizer operations are used to distill pure resource states out of the initially noisy resources.
However, not every nonstabilizer resource can be distilled in this fashion to generate a state which promotes the stabilizer subtheory to universal quantum computation. It is a major open question to determine which states are in fact sufficient for this purpose.

Quasiprobability representations are a central tool for making progress on these and related problems. For finite-dimensional quantum systems, a number of quasiprobability representations have been studied. For example, Gibbons, Hoffman, and Wootters (GHW) identified a family of representations on a discrete phase space~\cite{gibbons2004discrete}, and Gross then singled out one of these with a higher degree of symmetry~\cite{gross2006hudson}, by virtue of satisfying a property known as ``Clifford covariance''. All of these have been used to study quantum computation~\cite{PhysRevA.71.042302,PhysRevLett.109.230503,veitch2012negative,veitch2014resource,howard2014contextuality,PhysRevA.95.052334,PhysRevX.5.021003,PhysRevA.83.032310}. 

Gross's representation in particular has been the most useful in understanding the resources required for computation.
For instance, Ref.~\cite{veitch2012negative} extended the Gottesman-Knill theorem~\cite{gottesman1998heisenberg} by devising an explicit simulation protocol for quantum circuits composed of Clifford gates supplemented with arbitrary states and measurements that have nonnegative Gross's representation. Ref.~\cite{veitch2012negative} also proved that every state which is useful for magic state distillation necessarily has negativity in its Gross's representation. In Ref.~\cite{howard2014contextuality}, this result was leveraged to prove that every state that promotes the stabilizer subtheory to universal quantum computation via magic state distillation must also exhibit Kochen-Specker contextuality~\cite{KS}. In recognition that negativity in Gross's representation is a resource for quantum computation in this sense, Ref.~\cite{veitch2014resource} introduced an entire resource theory~\cite{coecke2016mathematical} of Gross's negativity. 

From a foundational perspective, it is surprising that any {\em particular} quasiprobability representation plays such a central role. As argued in Ref.~\cite{negativity}, negativity of any one quasiprobability representation is not sufficient to establish nonclassicality in general scenarios. So how can it be that Gross's representation plays such an important role, e.g. that negativity in it is associated to a strong form of nonclassicality, namely computational speedups?
Early clues were provided by Gross~\cite{gross2006hudson} and by Zhu~\cite{zhu}, each of whom proved that Gross's representation was the unique representation with some natural symmetry properties. However, it has previously been unclear what these properties have to do with nonclassicality, and both Gross's and Zhu's arguments relied on auxiliary mathematical assumptions that were not physically motivated (as we discuss below).

In this paper, we resolve this mystery by showing that the {\em only} nonnegative and diagram-preserving~\cite{schmid2020structure} quasiprobability representation of the stabilizer subtheory in any odd dimension is Gross's.
We also prove that in all even dimensions  (where Gross's representation is not defined), there is {\em no} nonnegative and diagram-preserving quasiprobability representation of the stabilizer subtheory. This implies that the stabilizer subtheory exhibits generalized contextuality in all even dimensions. 

In the setting of the full stabilizer subtheory, our result for odd dimensions proves that negativity {\em of this particular} quasiprobability representation is a rigorous signature of nonclassicality, i.e., the failure of generalized noncontextuality. \emph{Generalized noncontextuality} is a principled~\cite{Spe05,schmid2019characterization,SpekLeibniz19}, useful~\cite{schmid2018contextual,AWV,KLP19,POM,RAC,RAC2,Saha_2019,cloningcontext,comp1,comp2, YK20,saha2019preparation,PhysRevLett.125.230603}, and operational~\cite{operationalks,unphysical, Kunjwal16,Schmid2018,pusey2019contextuality,selby2021accessible,selby2022open} notion of classicality.
If one's process has negativity in Gross's representation, then our result establishes that there is no nonnegative representation of the full stabilizer subtheory together with that process. Since nonnegative quasiprobability representations are in one-to-one correspondence with generalized noncontextual ontological models~\cite{negativity,schmid2019characterization,schmid2020structure}, this means that there is no noncontextual representation for the scenario, and hence no classical explanation of it\footnote{Note that Ref.~\cite{schmid2020unscrambling} introduced a more refined framework for studying ontological models and noncontextuality, and argued that better terminology for these are `classical realist representations' and `Leibnizianity', respectively. We do not use this framework or terminology here only so that our results are easier to parse for readers who have not read Ref.~\cite{schmid2020unscrambling}.}. 

Our work also extends the body of known connections between contextuality, negativity, and computation~\cite{howard2014contextuality,veitch2014resource,AndersBrowne,PhysRevA.95.052334,frembs2018contextuality,PhysRevX.5.021003,PhysRevLett.119.120505, Delfosse2017,haferkamp2021equivalence,booth2021contextuality}.
Using known links between resources for quantum computation and negativity in Gross's representation, together with our result connecting such negativity to the failure of generalized noncontextuality, one can derive connections between resources for quantum computation and generalized noncontextuality.

We illustrate this by giving an analogue of the celebrated result in Ref.~\cite{howard2014contextuality}: namely, we prove that generalized contextuality is necessary for universal quantum computation in the state injection model.

Finally, we note that our main result demonstrates that the principle of generalized noncontextuality is a much stronger principle than was previously recognized, at least in some settings. This is exemplified by the fact that
%For example, 
for stabilizer theories in odd dimensions, it does not merely provide constraints on ontological representations, it {\em completely fixes} the ontological representation.  This offers some hope that if the notion of a generalized noncontextual model can be relaxed in such a way~\cite{schmid2020unscrambling} that lifts the obstructions to modelling the entirety of quantum theory, such a model of the full theory might also be unique. In our view such a uniqueness result would offer a compelling reason to take the identified ontology seriously.

\emph{The stabilizer subtheory---} 
The stabilizer subtheory is one of the most important subtheories of quantum theory in the field of quantum information, playing an important role in quantum computing~\cite{gottesman1997stabilizer,gottesman1998heisenberg,bravyi2005universal,howard2014contextuality,PhysRevA.83.020302,PhysRevA.70.052328}, quantum error correction~\cite{gottesman1997stabilizer,gottesman1998heisenberg,riste2015detecting,Terhal2015,PhysRevLett.95.230504}, and quantum foundations~\cite{spekkens2016quasi,pusey2012stabilizer,catani2017spekkens,PhysRevA.98.052108,Lillystone2019,lillystone2019contextual}. We introduce it\footnote{ Note that an alternative definition of the stabilizer subtheory in prime power dimensions that appears in literature utilizes the finite field $\mathbb{F}_{p^k}$ of $p^k$ elements. For example, this is the definition of the stabilizer subtheory studied by Zhu~\cite{zhu}. The two definitions coincide for prime $d = p$, but are distinct for prime power $d = p^k$, $k \neq 1$. In this work, we will follow Gross's definition and concern ourselves only with the stabilizer subtheories defined using the residue field $\mathbb{Z}_d$ for all odd $d$.} briefly here, with more details in Appendix~\ref{app:stab}.

The stabilizer subtheory is built around the Clifford unitaries.
To define these, we first introduce the \emph{Weyl operators} (also called generalized Pauli operators).
Consider a $d$-dimensional quantum system with computational basis $\{ \ket{0}, \ldots, \ket{d-1} \}$.  Writing $\omega = \exp(\frac{2\pi i}{d})$, we define the translation operator $X$ and boost operator $Z$ via 
\begin{align}
    X\ket{x} &= \ket{x+1} &
    Z\ket{x} &= \omega^x\ket{x}.
\end{align}
Note that here and throughout, all arithmetic is within $\Z_d$, the integers modulo $d$.
The single-system Weyl operators are then defined as
$
   %W_{p,q} = \omega^{-2^{-1}pq} Z^p X^q,
   W_{p,q} =  Z^p X^q,
$
where $p,q\in \mathbb{Z}_d$. Note that these are often defined with an additional phase factor $\omega^{\gamma_{p,q}}$; however, 
% the choice of this phase is irrelevant for the definition of the stabilizer subtheory, so we set $\gamma_{p,q}$ to zero. 
the resulting operational theory is the same for any valid~\cite{Delfosse2017} phase choice, so we will set $\gamma_{p,q}$ to zero.
The Clifford unitaries are defined as unitaries which---up to a phase---map Weyl operators to other Weyl operators under conjugation.

The stabilizer subtheory for a single system in dimension $d$ is defined as the set of processes which can be generated by sequential composition of: i) pure states uniquely identified by being the simultaneous eigenstates of a given set of Weyl operators, ii) projective measurements in the spectral decomposition of the Weyl operators\footnote{Note that although the Weyl operators are not Hermitian operators, they {\em are} normal operators, and hence have a spectral decomposition, which implies one can carry out a projective measurement in the eigenbasis of each.}, and iii) Clifford unitary superoperators on the associated Hilbert space, as well as convex mixtures of such processes.

This construction is easily generalized to allow for parallel composition, that is, for systems made up of $n$ qu$d$its\footnote{To the authors' knowledge, parallel composition of systems of different dimensions has not been considered in the literature.}, by defining the multiparticle Weyl operators as tensor products of those defined above, and defining the multiparticle Clifford operators as unitary superoperators that preserve the multiparticle Weyl operators under conjugation;  see Ref.~\cite{gross2006hudson} for more details. 
An important feature is that in general the stabilizer subtheory defined by parallel composition of $n$ qu$d$its is not the same as the stabilizer subtheory defined by a single $d^n$ dimensional system---for instance, the latter generally has far fewer states~\cite{gross2006hudson}. Therefore, for a given dimension $D$ there may be multiple different stabilizer theories which could be associated to it, depending on whether one views it as a single monolithic system of dimension $D$ (which Gross calls the single-particle view), or views it as some tensor product of multiple qu$d$its (which Gross calls a multi-particle view).

%\section{Quasiprobability  representations}
%\section{Gross's representation}

\emph{Quasiprobability representations---}
A quasiprobability representation~\cite{ferrie2008frame,van2017quantum,schmid2020structure} is akin to a mathematical representation of quantum processes as stochastic processes on a sample space, except that the representation may take negative values. 
% To define such a representation for arbitrary compositional scenarios, one greatly benefits from a diagrammatic approach like that taken in Ref.~\cite{schmid2020structure,schmid2020unscrambling}. 
For the reasons laid out in Refs.~\cite{schmid2020structure,schmid2020unscrambling}, we are only interested in quasiprobability representations that satisfy the assumption of \emph{diagram preservation}~\cite{schmid2020structure,schmid2020unscrambling}---namely, that the representation commutes with sequential and parallel composition of processes. 
%of a composite process is equal to the composition of the representations of its component processes. 
This assumption is satisfied by most of the useful quasiprobability representations considered in the literature, including the standard (continuous-dimensional) Wigner function and Gross's representation.

The arguments of Ref.~\cite{schmid2020structure} imply that every diagram-preserving quasiprobability representation of a full dimensional subtheory\footnote{That is, in which the states span the quantum state space and the effects span the quantum effect space. Note that the stabilizer subtheory is such a theory, which can be seen by noting that the Weyl operators span the space of Hermitian operators, and hence, so do their eigenstates.} of quantum theory can be written as a minimal frame representation~\cite{ferrie2008frame}, i.e. one whose frame elements form a basis, as follows. One first associates to each system a basis $\{F_\lambda\}_{\lambda}$ for the real vector space of Hermitian operators, where
\begin{align}
\tr{F_\lambda} &= 1 \label{framesum2}.
\end{align}
Every basis has a unique dual basis, $\{D_\lambda\}_{\lambda}$, as proved in Appendix~\ref{prelim}, where 
\begin{align}
    \sum_\lambda D_\lambda &= \mathds{1}, &  \label{dualsum}
\tr{D_{\lambda'} F_\lambda}&=\delta_{\lambda\lambda'}.
\end{align}
In this representation, a completely-positive trace-preserving map~\cite{nielsen2001quantum,Schmidcausal} $\mathcal{E}$ is represented by a quasistochastic map defined by 
% the conditional quasiprobability distribution
\begin{equation}
    \xi_{\mathcal{E}}(\lambda'|\lambda) = \tr{D_{\lambda'}\mathcal{E}(F_{\lambda})}.
\end{equation} 
As special cases, the representations of a state $\rho$ and an effect $E$ are given by
\begin{align}\label{eq:QuasiRepState}
    \xi_{\rho}(\lambda) &= \tr{D_{\lambda} \rho}, &
    \xi_{E}(\lambda) &= \tr{F_{\lambda} E},
\end{align}
and the quantum probabilities are recovered as
\begin{equation}
    \tr{E \mathcal{E}(\rho)}=\sum_{\lambda',\lambda} \xi_{E}(\lambda') \xi_{\mathcal{E}}(\lambda'|\lambda) \xi_{\rho}(\lambda).
\end{equation}

A quasiprobability representation is said to be {\em nonnegative} if for every process $\mathcal{E}$, $0 \leq \xi_{\mathcal{E}}(\lambda'|\lambda) \leq 1$ for every $\lambda, \lambda'$. In this case, the representation is in one-to-one correspondence with a noncontextual ontological model~\cite{Spe05,schmid2020unscrambling}.

% Motivate diagram preservation. 
% -motivate special instances. refer to our other two papers which argue for it. refer to many-particle quasirepns? mention notions of `local' explainability, e.g. if you add a process at the end of the circuit, the fundamental description for what is happpening elsewhere in the circuit should be unchanged. mention which representations satisfy/fail this constraint.

%\subsection{Gross's representation}

%\ 

\emph{Gross's representation---}
The particular quasiprobability representation introduced by Gross~\cite{gross2006hudson} is for odd dimensional quantum systems and takes the sample space to be a phase space $V = \Z_d \times \Z_d$, and so its elements will be labelled by $a := (p,q)$, rather than $\lambda$. 
Hence, the basis operators in Gross's representation are indexed by $a\in V$, and we will denote them by $A_a$.

The basis operators in Gross's representation can be written in terms of the Weyl operators as follows:  %They are defined as
\begin{equation}
    \{A_a\}_a := \left\{ \frac1d \sum_b \omega^{-[a,b]} {W^G_b}^\dagger \right\}_{a},
\end{equation}
where Gross's Weyl operators $W^G_{p,q}$ are related to ours via
%\beq
 $   W^G_{p,q}:= \omega^{2^{-1}pq} W_{p,q}.$
%\eeq
%   which it is easy to see are also an orthonormal basis w.r.t. the rescaled Hilbert-Schmidt inner product.
   %the phase convention in defining the Weyl operators is taken to be $\gamma^G_{p,q} := 2^{-1}pq$.
   %\jnote{What is $n$? I guess it's the number of Weyl operators, which should be $d^2$?}
   %and 
%These operators are often termed phase space point operators, and denoted $A_a$. T
These operators form an orthogonal basis, and so the basis is essentially self-dual, so that both $\{F_\lambda\}$ and $\{D_\lambda\}$ are proportional to $\{A_a\}$, with $D_\lambda = \frac1d F_\lambda$. They moreover satisfy a number of useful properties (see, e.g., Lemma~29 of Ref.~\cite{gross2006hudson}) including a key feature of {\em translational covariance}~\cite{gross2006hudson} where:
\begin{equation} \label{eq:CliffordCovariance}
    W_{p',q'} A_{p,q} W^{\dagger}_{p',q'} =  A_{p+p',q+q'} \quad \forall p,q,p',q'.
\end{equation}
%for any Clifford unitary $C_{a,S}$.

%\section{Uniqueness of Gross's representation}
%\section{Main result}

%\

\emph{Main result---} 
Our main result is a complete characterization of the (non)classicality of the stabilizer subtheory in every finite dimension.

\begin{theorem} \label{mainthm}
    \leavevmode
    \begin{enumerate}[label=(\alph*)]
        \item For any stabilizer subtheory (single- or multi-particle) in {\bf odd} dimensions, the {\em unique} nonnegative and diagram-preserving quasiprobability representation for it is Gross's representation.
        \item For any stabilizer subtheory (single- or multi-particle) in {\bf even} dimensions, there is no nonnegative and diagram-preserving quasiprobability representation.
    \end{enumerate}
\end{theorem}

The proof is given in Appendix~\ref{thmproof}.

As shown in Ref.~\cite{spekkens2016quasi,catani2017spekkens}, Gross's representation is identical to Spekkens' epistemically restricted toy theory~\cite{spekkens2007evidence} for odd dimensions~\cite{spekkens2016quasi}.
Through the equivalences between various notions of classicality \cite{schmid2020structure}, our result can be stated in a number of ways. Perhaps the most natural equivalent statement of Theorem~\ref{mainthm} is the following:
For odd dimensions, the unique noncontextual representation of the stabilizer subtheory is Spekkens' epistemically restricted toy theory. For even dimensions, the stabilizer subtheory is contextual.

%There are several senses in which our uniqueness result Theorem~\ref{mainthm}(a) is stronger than that proven by Gross~\cite{gross2006hudson}. Most importantly, the principle of generalized noncontextuality is a well-established notion of classicality, while the notion of Clifford covariance is not.Additionally, our result starts from the very weak assumption of classical realism~\cite{schmid2020unscrambling}---that is, the ontological models framework---while Gross's result requires two additional assumptions beyond this, namely that the representation is on a $d \times d$ phase space and gives the correct marginal probabilities. In our approach, both of these are derived. Finally, our uniqueness result holds in all odd dimensions, while Gross's uniqueness result was proven only for odd prime dimensions. 

There are several senses in which our uniqueness result, Theorem~\ref{mainthm}(a), is stronger than that proven by Gross~\cite{gross2006hudson} or that proven by Zhu~\cite{zhu}. Most importantly, the principle of generalized noncontextuality is a well-established notion of classicality, while Gross's notion of Clifford covariance and Zhu's (weaker) notion of Clifford covariance are not. Additionally, our result starts from the very weak assumption of classical realism~\cite{schmid2020unscrambling}---that is, the ontological models framework---while Gross's and Zhu's results rely on additional assumptions which have not been given physical motivation. In particular, both Gross's and Zhu's arguments only single out Gross's representation if one assumes that one's representation is on a $d \times d$ phase space, and that it gives the correct marginal probabilities\footnote{ Zhu does not actually make the latter assumption, but as a consequence he does not uniquely single out Gross's representation; rather, he single out a family of representations in terms of the so-called Wigner-Wootters bases. But Gross's paper (the last paragraph in the proof of Theorem 23) shows explicitly that this assumption regarding marginal probabilities is exactly what is required to single out Gross's representation from among the Wigner-Wootters representations. }. In our approach, both of these are derived. Finally, our uniqueness result holds in all odd dimensions, while Gross's uniqueness result was proven only for odd prime dimensions, and Zhu's only for prime power dimensions.

Theorem~\ref{mainthm}(b) establishes that every stabilizer subtheory of even dimension exhibits contextuality. 
While this result has previously been claimed to be true, it had not in fact been proven (to our knowledge). For $d=2$, there are well-known proofs of contextuality, e.g. in Ref.~\cite{Lillystone2019}. It follows that every subtheory which contains all the processes in the qubit stabilizer subtheory is also contextual. 
However, it is not known whether every even-dimensional stabilizer subtheory contains the qubit stabilizer as a subtheory (see Ref.~\cite{gross2006hudson}),
%for details), 
and so the claim of Theorem~\ref{mainthm}(b) does not trivially follow in this manner.

% For example, consider a system with even dimension $4$. One could view such a system as a composite of two qubits; this is an instance of what Gross calls the `multi-particle' view. Alternatively, one could consider it as a single, monolithic system; Gross calls this the `single-particle' view. These two options lead one to distinct operational theories; e.g., Gross has shown that there are many more multi-particle stabilizer states than there are single-particle stabilizer states~\cite{gross2006hudson}. Each of these might sensibly be termed a `stabilizer subtheory of dimension $4$'. The subtheory defined by the multi-particle view above does in fact exhibit contextuality, since it contains the single-qubit stabilizer theory as a subtheory. The single-particle stabilizer subtheories in even dimensions $d>2$, however, do not contain the single-qubit stabilizer subtheory, and hence are not obviously contextual. 

% Our Theorem~\ref{mainthm}(b) establishes that every single-particle stabilizer subtheory in even dimensions is contextual. 

%\section{Generalized contextuality as a resource for quantum computation}

\emph{Generalized contextuality as a resource for quantum computation---}
The stabilizer subtheory is efficiently simulable~\cite{gottesman1998heisenberg}. However, if one supplements it with appropriate nonstabilizer states, one can achieve universal quantum computation through magic state distillation~\cite{bravyi2005universal}. 

Any state which promotes the stabilizer subtheory to universal quantum computation must have negativity in its Gross's representation~\cite{veitch2012negative}.
Ref.~\cite{howard2014contextuality} further showed that Kochen-Specker contextuality is necessary for universality in this model of quantum computation.

The key argument of Ref.~\cite{howard2014contextuality} was a graph-theoretic proof that if a state is negative in Gross's representation, then it admits a (state-dependent) proof of Kochen-Specker contextuality using only stabilizer measurements. Our main theorem, Theorem~\ref{mainthm}, is analogous, establishing that if a state is negative in Gross's representation, then it admits a proof of {\em generalized} contextuality.

Hence, we immediately arrive at a result akin to that of Ref.~\cite{howard2014contextuality}:  generalized contextuality is necessary for universality in the state injection model of quantum computation.
\begin{theorem} \label{compthm}
    Consider any state $\rho$ which promotes the stabilizer subtheory to universal quantum computation. There is no generalized noncontextual model for the stabilizer subtheory together with $\rho$.
\end{theorem}
We comment in Appendix~\ref{alternativeproofs} on two other routes to proving this theorem.

%\subsection{On the sufficiency of generalized contextuality for universal quantum computation}

\emph{On the sufficiency of generalized contextuality for universal quantum computation---}
Thus far we have focused on the necessity of contextuality for quantum computation. However, the fact that Gross's representation provides the unique noncontextual representation of the stabilizer subtheory may also be useful for discovering in what sense (if any) generalized contextuality is {\em sufficient} for quantum computation.

Without any caveats, generalized contextuality is clearly not sufficient for universal quantum computation. This can be seen by the example of the stabilizer subtheory in dimension $2$, which admits proofs of contextuality~\cite{Lillystone2019} and yet is efficiently simulable~\cite{gottesman1998heisenberg}.

Still, it is conceivable that there is a more nuanced sufficiency result relating contextuality and computation, e.g. by leveraging quantitative measures of generalized contextuality~\cite{marvian2020inaccessible,selby2022open} or by focusing on particular dimensions and models of quantum computation.  We now prove a related result (without explicit reliance on Theorem~\ref{mainthm}).

From Ref.~\cite{anwar2012qutrit,veitch2012negative}, we know that access to enough copies of any nonstabilizer pure state promotes the stabilizer subtheory to universal quantum computation. 
Similarly, access to enough copies of any nonstabilizer unitary promotes the stabilizer subtheory to universal quantum computation, since the Clifford unitaries together with any other unitary gate forms a universal gate set~\cite{anygatemakesStabUniversal,campbell2012magic}.

It is well known that every pure nonstabilizer state is negatively represented in Gross's representation~\cite{gross2006hudson}. Additionally, it is not hard to see that every nonstabilizer unitary gate is negatively represented in Gross's representation. By the universal gate set property~\cite{anygatemakesStabUniversal,campbell2012magic}, combining the positively represented Clifford gates with any given nonstabilizer unitary allows the approximation of any other unitary---including one that maps some pure stabilizer state to some pure nonstabilizer state. Since the stabilizer state is represented positively and the nonstabilizer state must be represented negatively in Gross's representation, the unitary mapping between them must have negativity in its Gross's representation, and hence so must the given nonstabilizer unitary used to construct it. Hence we obtain the following theorem: 
\begin{theorem}
A (necessary and) sufficient condition for any unitary or pure state to promote the stabilizer subtheory to universal quantum computation is that it be negatively represented in Gross's representation.
\end{theorem}
For the case of pure states, this result was pointed out in Refs.~\cite{anwar2012qutrit,veitch2012negative}.  Perhaps the most important open question that remains is whether an analogous sufficiency result holds for mixed states and generic quantum channels.

\emph{Conclusion---} We have proved that noncontextuality picks out a unique classical explanation for every stabilizer subtheory in odd dimensions, and that there is no noncontextual model for any stabilizer subtheory in even dimensions. We then proved that, as a consequence, generalized contextuality is a necessary resource for universal quantum computation in the state injection model. We expect these results connecting contextuality and negativity to continue to be useful for understanding the resources needed for quantum information processing.

\subsection{Acknowledgements}
D.S. thanks Robert W. Spekkens and Stephen Bartlett for early discussions conjecturing the uniqueness of Gross's representation, and thanks Philippe Allard Gu\'erin and Lorenzo Catani for useful discussions.
D.S. was supported by a Vanier Canada Graduate Scholarship. MFP is supported by the Royal Commission for the Exhibition of 1851. D.S. and J.H.S. are supported by the Foundation for Polish
Science through IRAP project co-financed by EU within Smart Growth Operational Programme
(contract no. 2018/MAB/5). This research was supported by Perimeter Institute for Theoretical Physics. Research at Perimeter Institute is supported in part by the Government of Canada through the Department of Innovation, Science and Economic Development Canada and by the Province of Ontario through the Ministry of Colleges and Universities. 

\bibliographystyle{apsrev4-2}
\bibliography{bibliography}

\appendix

\section{The stabilizer subtheory} \label{app:stab}

We here expand on the exposition of the stabilizer subtheory from the main text (with some redundancy for completeness).

The stabilizer subtheory is built around the Clifford group, whose elements will be referred to as Clifford unitaries.
To define these, we first introduce the \emph{Weyl operators} (also called generalized Pauli operators).
Consider a $d$-dimensional quantum system, and define the computational basis $\{ \ket{0}, \ldots, \ket{d-1} \}$ in its Hilbert space $\mathcal{H}$. Each basis element is labelled by an element of $\Z_d$\footnote{When $d$ is prime, $\Z_d$ has the structure of a finite algebraic field. For non-prime $d$, things are somewhat more complicated \cite{gross2006hudson}, but the results in this work still hold.}, which we refer to as the configuration space. Writing $\omega = \exp(\frac{2\pi i}{d})$, we define the translation operator $X$ and boost operator $Z$ via 
\begin{align}
    X\ket{x} &= \ket{x+1}\\
    Z\ket{x} &= \omega^x\ket{x}.
\end{align}
Note that here and throughout, all arithmetic is within $\Z_d$.
These can be viewed as discrete position and momentum translation operators, respectively, for a particle on a ring.
From these, the single-system Weyl operators are defined as
\begin{equation}
   %W_{p,q} = \omega^{-2^{-1}pq} Z^p X^q,
   W_{p,q} =  Z^p X^q,
\end{equation}
where $p,q\in \mathbb{Z}_d$. Note that these are often defined with an additional phase factor $\omega^{\gamma_{p,q}}$; however, 
% the choice of this phase is irrelevant for the definition of the stabilizer subtheory, so we set $\gamma_{p,q}$ to zero.   
the resulting operational theory is the same for any valid~\cite{Delfosse2017} phase choice, so we will set $\gamma_{p,q}$ to zero. (We  highlight this lack of dependence on the chosen phase by introducing the stabilizer subtheory using superoperators, for which any valid choice of phase cancels.) 

The Weyl operators are unitaries whose associated superoperators, $\mathcal{W}_{p,q}(\cdot):= W_{p,q}(\cdot)W_{p,q}^\dagger$, form a group with composition law 
\begin{align}\label{eq:weyl-comp}
    %W_{p,q} W_{p',q'} &= \omega^{2^{-1} \small\left[ \colv{p}{q}, \colv{p'}{q'}\right]} W_{p+p',q+q'},
    \mathcal{W}_{p,q} \mathcal{W}_{p',q'} &=  \mathcal{W}_{p+p',q+q'},
\end{align}
and inverse
\begin{equation} \label{eq:WeylInverse}
    \mathcal{W}_{p,q}^{-1} = \mathcal{W}_{p,q}^\dagger = \mathcal{W}_{-p,-q}.
\end{equation}
(Note that the Weyl operators themselves do not form a group as the above equations only hold up to a particular phase factor.)

It will be useful later to note that the Weyl operators are orthonormal with respect to a rescaled Hilbert-Schmidt inner product:
  \begin{equation} \label{WeylON}
        \frac{1}{d} \, \tr{W_{p,q} W_{p',q'}^{\dagger}} = \delta_{p,p'} \delta_{q,q'}.
    \end{equation}
    
The Clifford unitaries are defined as unitaries which---up to a phase---map Weyl operators to other Weyl operators under conjugation. Equivalently, their associated superoperators map Weyl superoperators to other Weyl superoperators under conjugation. That is, $\mathcal{U}$ is a Clifford unitary superoperator if for every $(p,q)$, one has
\begin{equation}
    \mathcal{U} \mathcal{W}_{p,q} \mathcal{U}^\dagger =  \mathcal{W}_{p',q'}.
\end{equation}
%for some $\phi, p', q'$ (which depend on $p$, $q$, and $U$). 

For a fixed dimension, the Clifford superoperators form a group generated by
%this 
the superoperators associated to the generalized phase gate $P$~\cite{farinholt2014ideal}
and the generalized Hadamard gate $H$, whose explicit form
\begin{equation} \label{Hadamarddefn}
    H \ket{x} = \frac{1}{\sqrt{d}}\sum_{k \in \Z_d} \omega^{xk}\ket{k}
\end{equation}
will be used below.

The stabilizer subtheory for a single system in dimension $d$ is defined as the set of processes which can be generated by sequential composition of: i) pure states uniquely identified by being the simultaneous eigenstates of a given set of Weyl operators, ii) projective measurements in the spectral decomposition of the Weyl operators\footnote{Note that although the Weyl operators are not Hermitian operators, they {\em are} normal operators, and hence have a spectral decomposition, which implies one can carry out a projective measurement in the eigenbasis of each.}, and iii) Clifford unitary superoperators on the associated Hilbert space, as well as convex mixtures of such processes.

This construction is easily generalized to allow for parallel composition, that is, for systems made up of $n$ qu$d$its\footnote{To the authors' knowledge, parallel composition of systems of different dimensions is at best highly nontrivial, and has not been considered in the literature.}, by defining the multiparticle Weyl operators as tensor products of those defined above, and defining the multiparticle Clifford operators as unitary superoperators that preserve the multiparticle Weyl operators under conjugation;  see Ref.~\cite{gross2006hudson} for more details. 
An important feature is that in general the stabilizer subtheory defined by parallel composition of $n$ qu$d$its is not the same as the stabilizer subtheory defined by a single $d^n$ dimensional system---for instance, the latter generally has far fewer states~\cite{gross2006hudson}. Therefore, for a given dimension $D$ there may be multiple different stabilizer theories which could be associated to it, depending on whether one views it as a single monolithic system of dimension $D$ (which Gross calls the single-particle view), or views it as some tensor product of multiple qu$d$its (which Gross calls a multi-particle view).

\section{Useful Preliminaries} \label{prelim}

It is well-known that a basis of a vector space uniquely defines a dual basis in the dual vector space (i.e. the space of functionals on the vector space). We will leverage this fact, but in a slightly different form: 
%with an extra step of construction, so that the dual basis is again in the original vector space:
\begin{lemma}\label{lemmaunique}
Given any basis $\{F_\lambda\}_\lambda$ for a $d^2$-dimensional real vector space $\herm$ of Hermitian operators on a Hilbert space $\mathcal{H}$, there is a unique set $\{D_\lambda\}_\lambda$ of $d^2$ Hermitian operators satisfying
\beq\label{eq:dualFrame1}
\mathsf{tr}(D_{\lambda'} F_\lambda)=\delta_{\lambda,\lambda'},
\eeq
and $\{D_\lambda\}_\lambda$ also forms a basis for $\herm$.
\end{lemma}
 \begin{proof} 
Consider any basis $\{ F_\lambda\}_\lambda$ of $\herm$. It uniquely specifies a basis $\{\mathcal{D}_\lambda\}_\lambda$ of the dual vector space $\herm^*$, where $\{\mathcal{D}_\lambda\}_\lambda$ are linear functionals satisfying $\mathcal{D}_{\lambda'}(F_\lambda)=\delta_{\lambda,\lambda'}$.\footnote{To see that this is unique, consider a linear functional $\mathcal{D}'_{\lambda'}$ satisfying $\mathcal{D}'_{\lambda'}(F_\lambda) = \delta_{\lambda,\lambda'}$ for all $\lambda$. Since a linear functional is fully specified by its action on a basis, $\mathcal{D}'_{\lambda'}$ is the exact same functional as $\mathcal{D}_{\lambda'}$.} Now, in order to obtain again a set of Hermitian operators $\{D_\lambda\}_\lambda$, we use the Riesz representation theorem~\cite{riesz1914demonstration}, which states that each of these functionals $\mathcal{D}_\lambda$ can be written as the Hilbert-Schmidt inner product with a unique Hermitian operator $D_\lambda$, namely
\beq
\mathcal{D}_{\lambda}(\cdot)= \tr{(\cdot) D_{\lambda}}.
\eeq
This picks out a unique basis $\{D_{\lambda}\}_\lambda$ which satisfies Eq.~\eqref{eq:dualFrame1}. 
\end{proof}

Note that the operators $\{F_\lambda\}_\lambda$ and $\{D_\lambda\}_\lambda$ are both in $\herm$. 
% Furthermore, note that the set $\{D_\lambda\}$ also uniquely specifies the set $\{F_\lambda\}$ via Lemma~\ref{lemmaunique}. Nevertheless, to distinguishing one from the other when needed, we often refer to $\{F_\lambda\}$ as the basis and $\{D_\lambda\}$ as the dual basis of $\herm$.
For a basis $\{F_\lambda\}_\lambda$, we refer to the set $\{D_\lambda\}_\lambda$ constructed using this lemma as the {\em dual basis}.

 Recall that a quasiprobability representation is \emph{diagram-preserving}~\cite{schmid2020structure,schmid2020unscrambling} if it commutes with sequential and parallel composition of processes. 
Another useful lemma we will require is the following.
 \begin{lemma} \label{lemrevers}
A nonnegative and diagram-preserving quasiprobabilistic representation of any unitary superoperator $\mathcal{U}(\cdot):=U(\cdot)U^\dagger$ is given by a permutation; that is, by a conditional probability distribution
\beq
\xi_\mathcal{U}(\lambda'|\lambda) = \delta_{\sigma_U(\lambda'),\lambda}
\eeq
for some permutation $\sigma_U: \Lambda \to \Lambda$.
\end{lemma}
 \begin{proof} 
By definition, a nonnegative quasiprobabilistic representation $\xi$ represents every unitary superoperator $\mathcal{U}$ as a stochastic map from $\Lambda$ to itself, so $\xi_{\mathcal{U}}$ and $\xi_{\mathcal{U}^\dagger}$ are stochastic maps. By diagram preservation, it holds that $\xi_{\mathcal{U} \mathcal{U}^\dagger} = \xi_{\mathcal{U}} \circ \xi_{\mathcal{U}^\dagger}$. But $\mathcal{U}\mathcal{U}^\dagger = \mathbb{1}$, and hence $\xi_{\mathcal{U}\mathcal{U}^\dagger} = \xi_{\mathbb{1}}$, where (by diagram preservation) $\xi_{\mathbb{1}}$ must be the identity matrix. Therefore  $\xi_{\mathcal{U}^\dagger} \circ \xi_{\mathcal{U}}$ is the identity matrix, so $\xi_{\mathcal{U}^\dagger}$ is the left inverse of $\xi_\mathcal{U}$, and so (by the fact that they are square matrices) $\xi_\mathcal{U}$ and $\xi_{\mathcal{U}^\dagger}$ are inverses. But the only (square) stochastic matrices whose inverses are stochastic are permutations. Hence $\xi_\mathcal{U}$ is a permutation for every unitary $U$.
 \end{proof}

A final useful lemma is a well-known result from Ref.~\cite{Spe05}:
\begin{lemma}\label{ODSM}
Projective measurements have an outcome-deterministic representation in any noncontextual ontological model. That is, representation of the projectors in a projective measurement are conditional probability distributions valued in $\{0,1\}$. Furthermore, every ontic state is in the support of the representations of one and only one of eigenstates in any given projective measurement.
\end{lemma}

This lemma was originally proven for full quantum theory, but the proof is easily repeated within the stabilizer subtheory.

\section{Proof of Main Theorem} \label{thmproof}
Here we prove Theorem~\ref{mainthm}.
%\begin{paperthm} \label{mainthm}
%    \leavevmode
%    \begin{enumerate}[label=(\alph*)]
%        \item For any stabilizer subtheory (single- or multi-particle) in {\bf odd} dimensions, the {\em unique} nonnegative and diagram-preserving quasiprobability representation for it is Gross's representation.
%        \item For any stabilizer subtheory (single- or multi-particle) in {\bf even} dimensions, there is no nonnegative and diagram-preserving quasiprobability representation.
%    \end{enumerate}
%\end{paperthm}

We first give a one-paragraph intuitive proof sketch. Recall that the structure theorem of Ref.~\cite{schmid2020structure} gives a minimal frame representation as discussed in the main text. Starting with the single-particle case, we leverage the fact that noncontextuality implies outcome determinism to find a privileged labeling of the ontic states as points in phase space. We show that this labelling satisfies translational covariance. Using this and the fact that Weyl operators form a basis of the linear operators, we then show that the representation is fixed by the outcomes of measurements of Weyl operators on the $\lambda = (0,0)$ ontic state \footnote{We believe, but have not shown, that distinct GHW representations differ by exactly these choices of outcomes.}.
We give various conditions on these outcomes due to the Hermiticity of the phase point operator, the representation of the Hadamard, and from considering measurements of commuting pairs of Weyl operators. In odd dimensions, we show that the unique solution to these conditions is that which gives Gross's phase point operators. In even dimensions, we show that there is no solution. The generalization to multi-particle stabilizer subtheories is then shown to follow immediately.

We now give the full proof.

We start from the assumption that we have {\em some} nonnegative and diagram-preserving quasiprobability representation of the stabilizer subtheory in some finite dimension $d$. Note that this subtheory is tomographically local, and has GPT dimension $d^2$. Hence, Corollary VI.2 of Ref.~\cite{schmid2020structure} implies that the number of elements in the sample space is exactly $d^2$. Since a nonnegative and diagram-preserving quasiprobability representation is equivalent to a noncontextual ontological representation, we will refer to the elements of the sample space as `ontic states'.

The structure theorems in Ref.~\cite{schmid2020structure} (in particular, Corollary~VI.2) imply that this representation is a {\em minimal frame representation}~\cite{ferrie2008frame} composed of a basis $\{F_{\lambda}\}_{\lambda}$ and its dual $\{D_{\lambda}\}_{\lambda}$ (in the sense of Lemma~\ref{lemmaunique},) such that
the representation of a completely positive trace preserving map $\mathcal{E}$ is given by the conditional quasiprobability distribution
\beq \label{channelrepn}
\xi_{\mathcal{E}}({\lambda'}|\lambda) = \tr{D_{\lambda'}\mathcal{E}(F_\lambda)}.
\eeq
Here, 
$\{F_{\lambda}\}_{\lambda}$ is a spanning and linearly independent set of $d^2$ Hermitian operators, as is $\{D_{\lambda}\}_{\lambda}$, where these satisfy
\begin{align}
    \tr{F_\lambda} &= 1, \label{framesum}\\
    \sum_\lambda D_\lambda &= \mathds{1},
\end{align}
and
\beq
\tr{D_{\lambda'} F_\lambda}=\delta_{\lambda\lambda'}.
\eeq
(Note, however, that the elements of each basis need not be pairwise orthogonal.) Given an $\{F_\lambda\}$, the $\{D_\lambda\}$ satisfying these conditions are unique, so to specify a representation it suffices to determine the $\{F_\lambda\}$, as we will now do.

Consider in particular the two stabilizer measurements corresponding to the $X^{\dagger}$ and $Z$ operators.
If we label the outcome of $X^{\dagger}$ by $p \in \mathbb{Z}_d $ and the outcome of $Z$ by $q\in \mathbb{Z}_d$, then by outcome determinism (Lemma~\ref{ODSM}), each ontic state corresponds to an ordered pair $(p,q)$. In fact, this correspondence is bijective, and hence we can choose a useful labelling of the ontic states, i.e. $\lambda \mapsto (p,q)$ (so that measurements of $X^{\dagger}$ reveal $p$ and measurements of $Z$ reveal $q$).  
To see that the correspondence is surjective, consider an eigenstate of $X$ with eigenvalue $\omega^{-p_1}$. The ontic states in the support of its representation must have $p=p_1$ so that the outcome of an $X^\dagger$ measurement is always $p_1$. Furthermore, a measurement of $Z$ on this eigenstate gives a uniformly random outcome $q$, and so the ontic states in the support of its representation must include {\em every} ontic state of the form $(p_1,q)$, for arbitrary $q\in\Z_d$. This holds for all $d$ eigenstates of $X$, and thus for all $p_1 \in \Z_d$. 
So for every pair $p,q$, there exists some ontic state (in the support of one of the eigenstates of $X$) which has $(p,q)$ as its label.  This establishes surjectivity.
Since the number ($d^2$) of ontic states is the same as the number of pairs $(p,q)$,  surjectivity implies bijectivity.

Next, we show that the assumed labelling forces the representation to manifestly satisfy translational covariance: that is, the unitary superoperator $\mathcal{W}_{p_1,q_1}(\cdot) := Z^{p_1}X^{q_1}(\cdot)\left(X^{q_1}\right)^\dagger \left(Z^{p_1}\right)^\dagger$ is represented by the permutation $(p,q) \rightarrow (p+p_1,q+q_1)$.
To see this, first recall that the representation of this unitary superoperator is necessarily a permutation, as shown in Lemma~\ref{lemrevers}.
Next, we determine the representation of the unitary superoperator $\mathcal{X}(\cdot):= X(\cdot)X^\dagger$.  Consider an eigenstate of $X$ with eigenvalue $\omega^{-p_1}$. We argued above that the ontic states in the support of its representation must have $p=p_1$. Because the state is invariant under the unitary  superoperator  $\mathcal{X}$, the value of $p$ must be unchanged by it. Similarly, consider an eigenstate of $Z$ with eigenvalue $\omega^{q_1}$. The ontic states in the support of its representation must have $q=q_1$.
Applying the unitary superoperator $\mathcal{X}$ increments the $Z$ eigenstate and corresponding eigenvalue by one, so that the value of $q$ is transformed to $q_1+1$. Hence, we see that the representation of the unitary superoperator $\mathcal{X}$ takes $p \rightarrow p$ and $q\rightarrow q+1$, which fully specifies its action as a permutation on the ontic states. (Note that this argument only holds for ontic states in the support of one of the $X$ eigenstates and also in the support of one of the $Z$ eigenstates. But by Lemma~\ref{ODSM}, every ontic state is of this sort.) By a similar argument, the representation of the  unitary superoperator $\mathcal{Z}(\cdot):=Z(\cdot)Z^\dagger$ takes $p \rightarrow p+1$ and $q\rightarrow q$. Since all Weyl unitary superoperators can be generated by composing $\mathcal{X}$ and $\mathcal{Z}$, and since the representation is diagram-preserving, this fully specifies the permutations representing all of the Weyl unitary superoperators.  In particular, the unitary superoperator $\mathcal{W}_{p_1,q_1}$ is indeed represented by the permutation $(p,q) \mapsto (p+p_1,q+q_1)$.

By a similar argument, we can deduce the representation of the Hadamard unitary superoperator $\mathcal{H}(\cdot):=H(\cdot)H^\dagger$, where $H$ is defined in Eq.~\eqref{Hadamarddefn}. In particular, if we start in the eigenstate of $X$ with eigenvalue $\omega^{p_1}$, then $p = -p_1$, and the Hadamard maps this to the eigenstate of $Z$ with eigenvalue  $\omega^{p_1}$, for which $q = p_1$. So we see that the permutation representing the Hadamard superoperator 
results in a final value for $q$ equal to the initial value of $-p$.
Similarly, for the eigenstate of $Z$ with eigenvalue $\omega^{q_1}$, one has $q = q_1$, and this is mapped to the eigenstate of $X$ with eigenvalue $\omega^{-q_1}$, for which $p = q_1$. So we see that the permutation representing the Hadamard superoperator also
% takes $q$ to $p$. 
results in a final value for $p$ equal to the initial value of $q$.
This fully specifies its action as a permutation on the ontic states, namely $(p,q) \mapsto (q,-p)$\footnote{In odd dimension, this is the Clifford covariant representation.}. In particular, $(0,0) \mapsto (0,0)$.

% Note that there is no such argument for the phase gate $S$. However, perhaps surprisingly, we will not require  

% (e.g., if the offset is $v\in \mathbb{Z}$, then the measurement that would reveal $q + p$ according to Gross's representation would instead reveal $q + p + v$). [David: not super clear, so I cut it]

 Recall that the Weyl operators are defined as 
\begin{equation}
    W_{p,q} = Z^p X^q.
\end{equation}
Now, since these operators (or their conjugates) are a basis for the complex vector space of linear operators on the Hilbert space, we can decompose the operator $F_{0,0}$ (namely, the element of the basis $\{F_{\lambda} \}_\lambda$ with $\lambda = (0,0)$) as
\beq \label{expandF00}
F_{0,0} = \frac1d \sum_{p,q} f_{p,q} W^\dagger_{p,q}
\eeq

Consider a measurement of a given Weyl operator $W_{p_1,q_1}$ when the ontic state happens to be $(0,0)$. By outcome determinism (Lemma \ref{ODSM}), we will always get a particular outcome, which we will label $v_{p_1,q_1}$. $W_{p_1,q_1}$ has a spectral decomposition $\sum_{ \alpha} \omega^\alpha \Pi^{p_1,q_1}_\alpha$ in terms of its eigenvalues $\omega^\alpha$ for $\alpha \in \mathbb{R}$ and the projectors $\Pi^{p_1,q_1}_\alpha$ onto the corresponding eigenvectors. Computing the quantity $\tr{F_{0,0}W_{p_1,q_1}}$, we obtain
\begin{equation} \label{doublequant}
    \tr{F_{0,0}W_{p_1,q_1}} = \sum_\alpha \omega^\alpha \tr{F_{0,0}\Pi^{p_1,q_1}_\alpha}.
\end{equation}
But we know that $\tr{F_{0,0}\Pi^{p_1,q_1}_\alpha}$ is the probability of outcome $\alpha$ occurring in a measurement of $W_{p_1,q_1}$ when the ontic state is $(0,0)$, and we have already defined that the outcome that must occur in this case is that corresponding to eigenvalue $\omega^{v_{p_1,q_1}}$.
% by outcome determinism (Lemma~\ref{ODSM}) we know that every term in this sum is zero except for a single term, which takes value one. 
It follows that $\tr{F_{0,0}W_{p_1,q_1}} =\omega^{v_{p_1,q_1}}$. 

But a substitution of Eq.~\eqref{expandF00} into the left-hand side of Eq.~\eqref{doublequant} also allows us to compute this value as
\begin{equation} \label{computeoneway}
    \tr{F_{0,0}W_{p_1,q_1}} = \frac1d \sum_{p,q} f_{p,q} \tr{W^\dagger_{p,q} W_{p_1,q_1}} = f_{p_1,q_1},
\end{equation}
where the last equality follows from Eq.~\eqref{WeylON}.
Hence $f_{p,q}=\omega^{v_{p,q}}$, and so
\begin{equation} \label{newexpandF00}
    F_{0,0} = \frac1d \sum_{p,q}\omega^{v_{p,q}}W^\dagger_{p,q}.
\end{equation}

Using $X^q Z^p = \omega^{-pq}Z^p X^q$ we can calculate 
\begin{equation}
    W_{p,q}^\dagger = X^{-q}Z^{-p} = \omega^{-pq}Z^{-p}X^{-q} = \omega^{-pq}W_{-p,-q},\label{arbWdagger1}
\end{equation}
so that
\begin{equation}
    W_{p,q} = \omega^{pq}W_{-p,-q}^\dagger.\label{arbWdagger}
\end{equation}
We require $F_{0,0}$ to be Hermitian, i.e.
\begin{equation}
  \frac1d \sum_{p,q} \omega^{v_{p,q}}W^\dagger_{p,q} =  F_{0,0} = F^\dagger_{0,0} = \frac1d \sum_{p,q} \omega^{-v_{p,q}} W_{p,q}
\end{equation}
and so by using Eq.~\eqref{arbWdagger} equating the phases in front of $W^\dagger_{p,q}$ this becomes
\begin{equation}
v_{p,q} = -v_{-p,-q} + pq.\label{arbWherm}
\end{equation}
Since $\omega^d = 1$ we actually only have this equation mod $d$, but we leave this implicit in any equations involving $v_{p,q}$.

For the Hadamard we have
\begin{multline}
HW^\dagger_{p,q}H^\dagger = HX^{-q}Z^{-p}H^\dagger = Z^{-q}X^{p} = \omega^{-pq}X^p Z^{-q} \\= \omega^{-pq}W^\dagger_{q,-p}.
\end{multline}

Hence $(0,0) \xmapsto{H} (0,0)$, i.e. $F_{0,0} = H F_{0,0} H^\dagger$, becomes
\begin{equation}
v_{p,q} = v_{-q,p} + pq. \label{arbWcovH}
\end{equation}
Applying this twice gives
\begin{equation}
    v_{p,q} = v_{-p,-q}.\label{arbWcovHH}
\end{equation}
Summing Eqs.~\eqref{arbWherm} and \eqref{arbWcovHH} gives
\begin{equation}
    2v_{p,q} = pq.\label{arbWcondition}
\end{equation}

Now consider a pair of commuting Weyl operators $W_{p,q}$ and $W_{p',q'}$, where the requirement that they commute can be expressed as $pq' - qp'=0$. They are jointly measurable, and give outcomes $v_{p,q}$ and $v_{p',q'}$ on the $(0,0)$ ontic state. Their product $W_{p,q}W_{p',q'}$ is also jointly measurable with both. It is a general feature of quantum theory that if measurements of some commuting $A$ and $B$ give eigenvalues $a$ and $b$, a measurement of their product $AB$ gives eigenvalue $ab$. Here we have $a = \omega^{v_{p,q}}$ and $b = \omega^{v_{p',q'}}$ so the outcome of $W_{p,q}W_{p'q'}$ on the $(0,0)$ ontic state must also be $v_{p,q} + v_{p',q'}$. But
\begin{multline}
    W_{p,q}W_{p'q'} = Z^p X^q Z^{p'} X^{q'} \\= \omega^{-p'q} Z^{p+p'} X^{q+q'} = \omega^{-p'q} W_{p+p',q+q'}.
\end{multline}

Since the outcome of $W_{p+p',q+q'}$ on $(0,0)$ is $v_{p+p',q+q'}$, this gives the outcome of $W_{p,q}W_{p'q'}$ as $v_{p+p',q+q'} - p'q$. But we already established that this outcome must be $v_{p,q} + v_{p',q'}$, so that
\begin{equation}
    v_{p+p',q+q'} = v_{p,q} + v_{p',q'} + p'q \label{arbWsum}
\end{equation}
for all such $(p,q,p',q')$. 

In the special case when $p' = p$ and $q'=q$ the commutation condition is clearly satisfied, and hence
\begin{equation}
    v_{2p,2q} = 2v_{p,q} + pq.
\end{equation}
Then we can apply Eq.~\eqref{arbWcondition} to obtain
\begin{equation}
    v_{2p,2q} = 2pq. \label{arbWevenpq}
\end{equation}

We now consider three cases, depending on the dimension $d$.

\subsection{Odd $d$}

In odd $d$ we have $W_{p,q}^d = (Z^pX^q)^d = \mathds{1}$ \cite{hostens}, so that the eigenvalues of $W_{p,q}$ are $d$-th roots of unity. Hence the $v_{p,q} \in \Z_d$.
In odd $d$, $\Z_d$ contains a unique inverse of $2$ so we can multiply each side of Eq.~\eqref{arbWcondition} by $2^{-1}$ to obtain the unique solution
\begin{equation}
    v_{p,q} = 2^{-1}pq\label{arbWgross}
\end{equation}
Hence $F_{0,0} = \frac1d \sum_{p,q} \left(\omega^{-2^{-1}pq}W_{p,q}\right)^\dagger$ is Gross's phase point operator.

Furthermore, we already argued that our representation must satisfy translation  covariance, which is satisfied if and only if $F_{p,q} = W_{p,q} F_{0,0} W^\dagger_{p,q}$; since Gross's representation also satisfies translation covariance its $F_{p,q}$ are likewise. Hence, the set of basis operators $\{F_\lambda \}_\lambda=\{F_{p,q} \}_{p,q}$ is exactly equal to the set of phase point operators in Gross's representation. 

Hence, any nonnegative and diagram-preserving quasiprobability representation for the stabilizer subtheory in odd dimensions is equivalent to Gross's.

\subsection{Even $d$, not a multiple of $4$}
In even $d$, there are values of $p,q$ for which $(Z^pX^q)^d \neq \mathds{1}$ \cite{hostens}, so the above argument for $v_{p,q} \in \Z_d$ is not applicable. We do have $(Z^p X^q)^{2d} = \mathds{1}$, so the eigenvalues are $2d$-th roots of unity, which can be represented in our convention by allowing half-integer $v_{p,q}$. For simplicity, we allow arbitrary $v_{p,q} \in \mathbb{R}$ in the following, remembering that the $v_{p,q}$ only appear as exponents of $\omega$ and so any equations involving them are still mod $d$.

If $d$ is even but not a multiple of $4$ then we can write $d = 2h$ where $h$ is odd and $h = -h$ mod $d$. If we set $p = q = h$ then the Hadamard covariance condition in Eq.~\eqref{arbWcovH} becomes
\begin{equation}
    v_{h,h} = v_{h,h} + h^2
\end{equation}
so that $h^2 = 0 \mod d$. But as $h$ is odd, we have $h = 1 \mod 2$ and so, multiplying this equation by $h$ and using that $2h=d$ we find that $h^2 = h \mod 2h  = h \mod d$.  We therefore have a contradiction, and hence there are no valid models in this case. 

\subsection{$d$ a multiple of $4$}
The remaining case is that $d$ is a multiple of $4$, i.e. $d = 4r$ for some non-zero $r$. If we set $(p,q) = (0,2)$ and $(p',q') = (2r,2(r-1))$ then $pq' - qp' = -4r = -d = 0$ so that we can apply Eq.~\eqref{arbWsum} to obtain
\begin{equation}
    v_{2r,2r} = v_{0,2} + v_{2r,2(r-1)}.
\end{equation}
Applying Eq.~\eqref{arbWevenpq} to each term this becomes
\begin{equation}
    2r^2 = 0 + 2r(r-1),
\end{equation}
so that $2r = 0$. But $2r \neq 0$, so there is no valid model in this $d$ either. Together with the previous case this establishes there are no valid models in any even dimension.

\subsection{Multipartite cases}

The multipartite generalization of these results follows immediately from the fact that the assumption of diagram preservation applies to parallel composition (not only sequential composition).
More explicitly, it follows from Proposition VI.6 of 
Ref.~\cite{schmid2020structure}, which implies 
that the frame representation for processes on a pair of systems is uniquely determined by the frame representation for processes on each component system. 
In the case that the component systems are odd-dimensional, they each have a unique representation, and hence, so too does the composite system. 
In the case that the component systems are even-dimensional, they do not admit of any noncontextual representation, and hence, neither does the composite system.

\section{Alternative arguments for the necessity of generalized contextuality} \label{alternativeproofs}
%Recall from the main text
%\begin{paperthm} \label{compthm}
%    Consider any state $\rho$ which promotes the stabilizer subtheory to universal quantum computation. There is no generalized noncontextual model for the stabilizer subtheory together with $\rho$.
%\end{paperthm}
Recall Theorem~\ref{compthm} from the main text: for any state $\rho$ which promotes the stabilizer subtheory to universal quantum computation, there is no generalized noncontextual model for the stabilizer subtheory together with $\rho$.

One might expect that this result follows immediately from the fact that there is no nonnegative quasiprobability representation of full quantum theory, and that such a proof would hold in every model of quantum computation. However, the mere fact that a universal quantum computer can {\em simulate} every quantum circuit does not necessarily imply that one can {\em implement} every quantum circuit. (The loophole here follows from the distinction between computational universality and strict universality~\cite{aharonov2003simple}. For example, the Toffoli and Hadamard gate together form a computationally universal gate set, and yet composition of these two gates cannot generate arbitrary unitary gates---only those with real matrix elements.) 
Hence, one cannot without further arguments conclude that a universal quantum computer is capable of implementing circuits with negativity (or contextuality)---one can only conclude that it can simulate such circuits.

 However, Theorem~\ref{compthm} can be proven by leveraging the previous necessity result for Kochen-Specker contextuality~\cite{howard2014contextuality} together with the fact that Kochen-Specker contextuality implies generalized contextuality~\cite{kunjwal2019beyond,kunjwal2018from}. This argument is not entirely immediate, insofar as the latter implication requires bringing auxiliary operational processes into the argument, and one must establish that all of these additional processes are within the stabilizer subtheory. However, this can be shown to be the case. 
First, one establishes outcome determinism for ontic states in the support of the maximally mixed state following the logic of Ref.~\cite{Spe05}, but using only stabilizer preparations. One then establishes that every ontic state in the support of the given nonstabilizer state (from the state-dependent proof of Ref.~\cite{howard2014contextuality}) is also in the support of the maximally mixed state, using the fact that there always exists a decomposition of the maximally mixed state into the given nonstabilizer state together with {\em only} stabilizer states.

\end{document}